\documentclass[11pt,letterpaper]{article}

\usepackage[margin=1in]{geometry}
\usepackage{latexsym,graphicx,amssymb}
\usepackage{theorem}
\usepackage{amsmath,enumerate}
\usepackage{wrapfig}
\usepackage{subfigure}
\usepackage{float}
\usepackage{epsfig}
\usepackage{xspace}
\usepackage{paralist}
\usepackage{enumerate}
\usepackage{cases}
\usepackage{caption}
\usepackage[linesnumbered,lined,boxed,commentsnumbered]{algorithm2e}
\usepackage{color}
\usepackage{complexity}

\usepackage{times}

\usepackage{multicol}
\newenvironment{proof}{\noindent {\bf Proof:  }}{}

\newcommand{\qed}{\hfill\rule{2mm}{2mm}}

\numberwithin{figure}{section}
\numberwithin{equation}{section}

\newtheorem{theorem}{Theorem}[section]

\newtheorem{corollary}[theorem]{Corollary}

\newtheorem{lemma}[theorem]{Lemma}

\newtheorem{fact}[theorem]{Fact}

{
\theoremstyle{definition}
\newtheorem{definition}[theorem]{Definition}

}

\newcommand{\ignore}[1]{}

\newcommand{\SDP}{\ensuremath{\mathsf{SDP}}\xspace}
\newcommand{\Dual}{\ensuremath{\mathsf{Dual}}\xspace}
\newcommand{\oracle}{\textsc{Oracle}\xspace}

\newcommand{\Te}{\mathsf{T}_e}
\newcommand{\He}{\mathsf{H}_e}
\newcommand{\tail}{\mathsf{T}}
\newcommand{\head}{\mathsf{H}}

\newcommand{\transpose}{\ensuremath{\mathsf{T}}}
\newcommand{\norm}[1]{\left\lVert#1\right\rVert}

\newcommand{\one}{\ensuremath{\mathbf{1}}}
\newcommand{\spn}{\ensuremath{\mathsf{span}}}
\newcommand{\ws}{\mathfrak{W}}
\newcommand{\mf}[1]{\ensuremath{\mathfrak{#1}}}
\newcommand{\tri}[3]{\ensuremath{\begin{bmatrix}
\{#1, #3\} \\
#2
\end{bmatrix}}}

\newcommand{\ceil}[1]{\ensuremath{\lceil #1 \rceil}}

\newcommand{\ra}{\rightarrow}

\renewcommand{\R}{\mathbb{R}}
\newcommand{\Z}{\mathbb{Z}}

\newcommand{\xcite}[1]{\cite{#1}}

\newcommand{\Tr}{\mathbf{Tr}}
\newcommand{\T}{\ensuremath{\top}}
\newcommand{\bm}[1]{\ensuremath{\mathbf{#1}}}

\newcommand*\samethanks[1][\value{footnote}]{\footnotemark[#1]}

\title{An SDP Primal-Dual Approximation Algorithm for Directed Hypergraph Expansion and Sparsest Cut with Product Demands}

\author{T-H. Hubert Chan\thanks{Department of Computer Science, the University of Hong Kong. {\texttt{hubert@cs.hku.hk, btsun@connect.hku.hk}}} \and Bintao Sun\samethanks
 }
\date{}

\begin{document}

\begin{titlepage}

\maketitle

\begin{abstract}
We give approximation algorithms for the edge expansion and sparsest cut with product demands problems on directed hypergraphs, which subsume previous graph models such as undirected hypergraphs and directed normal graphs.

Using an SDP formulation adapted to directed hypergraphs, we apply the SDP primal-dual framework by Arora and Kale (JACM 2016) to design polynomial-time algorithms whose approximation ratios match those of algorithms previously designed for more restricted graph models.  Moreover, we have deconstructed their framework and simplified the notation to give a much cleaner presentation of the algorithms.

\end{abstract}

\thispagestyle{empty}
\end{titlepage}

\section{Introduction}
\label{sec:intro}

The edge expansion of an edge-weighted graph gives a lower bound
on the ratio of the weight of edges leaving any subset~$S$ of vertices to the
sum of the weighted degrees of $S$.
Therefore, this notion has applications in graph partitioning or clustering~\cite{jacm/KannanVV04,colt/MakarychevMV15,PengSZ15}, in which a graph is partitioned into clusters such that, loosely speaking,
the general goal is to minimize the number of edges crossing different clusters with respect to some notion of cluster weights.

The edge expansion and the sparsest cut problems~\cite{LeightonR99} can be viewed as a special case when the graph is partitioned into two clusters.
Even though the involved problems are NP-hard, approximation algorithms have been developed for them in various graph models and settings, such as
undirected~\cite{jacm/AroraRV09,arora2008euclidean} or directed graphs~\cite{stoc/AgarwalCMM05,stoc/AgarwalAC07}, and uniform~\cite{jacm/AroraRV09,stoc/AgarwalCMM05} or general demands~\cite{stoc/AgarwalAC07,arora2008euclidean} in the case of sparsest cut.
Recently, approximation algorithms have been extended to the case
of undirected hypergraphs~\cite{louis2014approximation}.
In this paper, we consider these problems for the even more general class
of directed hypergraphs.


\ignore{
On the other hand, the celebrated
Cheeger's inequality~\cite{alon1985lambda1, alon1986eigenvalues} gives
upper and lower bounds of the edge expansion $\phi_G$ in terms
of the second eigenvalue $\gamma_2$ of the graph Laplacian as follows:
$$\frac{\gamma_2}{2} \leq \phi_G \leq \sqrt{2\gamma_2}.$$

Recently, higher-order Cheeger inequalities~\cite{lee2014multiway,kwok2015improved} have been achieved
to relate higher order eigenvalues of the graph Laplacian
with multi-way edge expansion, and conditions in which spectral clustering
can be applied have been explored~\cite{PengSZ15}.

For more general graph models,
Cheeger's inequality has been extended
to undirected hypergraphs~\cite{louis2015hypergraph,chan2018jacm}
and directed normal graphs~\cite{yoshida2016nonlinear}.
The high level approach is to define a diffusion operator
$\L_\omega: \R^V \ra \R^V$ and use properties of the
diffusion process $\frac{d f}{d t} = - \L_\omega f \in \R^V$
to achieve variants of the Cheeger's inequality.
Moreover, an approximation algorithm for edge expansion
has been developed for undirected hypergraphs~\cite{louis2014approximation}.

In this paper, we develop a formal spectral framework to analyze
edge expansion and related problems in directed hypergraphs~\cite{gallo1993directed}, which
are general enough to subsume all previous graph models.
}

\noindent \textbf{Directed Hypergraphs.}
We consider an edge-weighted \emph{directed  hypergraph} $H=(V,E,w)$, where~$V$ is the vertex set of size~$n$
and~$E\subseteq 2^V \times 2^V$ is the set of $m$ directed hyperedges;
Each directed hyperedge $e\in E$ is denoted by $(\Te, \He)$, where $\Te\subseteq V$ is the \emph{tail} and $\He\subseteq V$ is the \emph{head};
we assume that both the tail and the head are non-empty, and
we follow the convention that the direction is from tail to head.
We denote $r := \max_{e \in E} (|\Te| + |\He|)$.

The function $w: E \ra \R_+$ assigns a non-negative weight to each edge.
Note that $\Te$ and $\He$ do not have to be disjoint.
This notion of directed hypergraph was first introduced
by Gallol et al.~\cite{gallo1993directed},
who considered applications in propositional logic, analyzing dependency in relational database, and traffic analysis.

Observe that this model captures previous graph models:
(i) an undirected hyperedge $e$ is the special case when $\Te = \He$, and (ii) a directed normal edge $e$ is the special case when $|\Te| = |\He| = 1$.

\noindent \textbf{Directed Hyperedge Expansion.}
In addition to edge weights,
each vertex $u \in V$ has weight $\omega_u := \sum_{e \in E: u \in \Te \cup \He} w_e$ that is also known as its \emph{weighted degree}.
Given a subset $S \subseteq V$, denote $\overline{S} := V \setminus S$ and $\omega(S) := \sum_{u\in S}\omega_u$.
Define the out-going cut $\partial^+(S) := \{e \in E:  \Te \cap S \neq \emptyset \wedge \He \cap \overline{S} \neq \emptyset\}$,
and the in-coming cut $\partial^-(S) :=
\{e \in E:  \Te \cap \overline{S} \neq \emptyset \wedge \He \cap S \neq \emptyset\}$.
The out-going edge expansion of $S$ is $\phi^+(S) := \frac{w(\partial^+(S))}{\omega(S)}$,
and the in-coming edge expansion is $\phi^-(S) := \frac{w(\partial^-(S))}{\omega(S)}$.
The edge expansion of $S$ is $\phi(S) := \min \{\phi^+(S), \phi^-(S)\}$.
The edge expansion of $H$ is
$$\phi_H := \min_{\emptyset \neq S \subset V:
\omega(S) \leq \frac{\omega(V)}{2}} \phi(S).$$


\noindent \textbf{Directed Sparsest Cut with Product Demands.}
As observed in previous works such as~\xcite{jacm/AroraRV09},
we relate the expansion problem to
the sparsest cut problem with \emph{product demands}.
%
For vertices $i \neq j \in V$, we assume that the demand
between $i$ and $j$ is symmetric and given by the product
$\omega_i \cdot  \omega_j$.
For $\emptyset \neq S \subsetneq V$,
its \emph{directed sparsity} is $\vartheta(S) := \frac{w(\partial^+(S))}{\omega(S) \cdot \omega(\overline{S})}$.
The goal is to find a subset~$S$ to minimize $\vartheta(S)$.

Observe that $\omega(V) \cdot \vartheta(S)$ and $\frac{w(\partial^+(S))}{\min\{\omega(S) , \omega(\overline{S})\}}$
are within a factor of 2 from each other.  Therefore, the directed edge expansion problem
on directed hypergraphs can be reduced (up to a constant factor) to the sparsest cut problem with product demands.  Hence, for the rest of the paper,
we just focus on the sparsest cut problem with product demands.

\noindent \emph{Vertex Weight Distribution.} For the sparsest cut problem,
the vertex weights $\omega: V \ra \R_+$ actually do not have to be related
to the edge weights.  However, we do place restrictions on the
\emph{skewness} of the weight distribution.
Without loss of generality, we can assume that each vertex has integer weight.
For $\kappa \geq 1$, the weights $\omega$ are $\kappa$-skewed,
if for all $i \in V$, $1 \leq \omega_i \leq \kappa$.
In this paper, we assume $\kappa \leq n$.


\noindent \textbf{Balanced Cut.}
For $0 < c < \frac{1}{2}$, a subset $S \subseteq V$ is $c$-balanced if both $\omega(S)$
and $\omega(V \setminus S)$ are at least $c \cdot \omega(V)$.

\ignore{

In this paper, we study directed hyperedge expansion from the following two perspectives:

\begin{compactitem}

\item We give an algorithm to approximate the edge expansion~$\phi_H$ of a given directed hypergraph~$H$
and return a corresponding subset $S$ with small expansion $\phi(S)$,
which can also be used for estimating the second eigenvalue of the diffusion operator.
\end{compactitem}
}

\subsection{Our Contributions and Results}

Our first observation is a surprisingly simple reduction of the
problem from the more general directed hypergraphs to
the case of directed normal graphs.

\begin{fact}[Reduction to Directed Normal Graphs]
\label{fact:reduction}
Suppose $H = (V,E)$ is a directed hypergraph
with edge weights $w$ and vertex weights $\omega$.
Then, transformation to a directed normal graph
$\widehat{H} = (\widehat{V}, \widehat{E})$,
where $|\widehat{V}| = n + 2m$
and $|\widehat{E}| =m + \sum_{e \in E} (|\Te| + |\He|)$,
  is defined as follows.

\noindent The new vertex set is
$\widehat{V} := V \cup \{v_e^\tail, v_e^\head: e \in E \}$,
i.e., for each edge $e \in E$, we add two new vertices;
the old vertices retain their original weights, and the new vertices
have zero weight.

\noindent The new edge set is $\widehat{E} := \{(v_e^\tail, v_e^\head): e \in E\}
\cup \{(u, v_e^\tail): e \in E, u \in \Te\} \cup
\{(v_e^\head, v): e \in E, v \in \He\}$.
An edge of the form $(v_e^\tail, v_e^\head)$ has its weight $w_e$
derived from $e \in E$, while all other edges
have large weights $\mathfrak{M} := n \sum_{e \in E} w_e$.

\noindent We overload the symbols for edge $w$ and vertex $\omega$ weights.
However, we use $\widehat{\partial}^+(\cdot)$ for out-going cut in $\widehat{H}$.

\noindent Given a subset $S \subseteq V$,
we define the transformed subset $\widehat{S} := S \cup \{v_e^\tail: S \cap \Te \neq \emptyset\} \cup \{v_e^\head: \He \subseteq S\}$.
Then, we have the following properties.

\begin{compactitem}
\item For any $S \subseteq V$,
$\omega(S) = \omega(\widehat{S})$ and $w(\partial^+(S)) = w(\widehat{\partial}^+(\widehat{S}))$.
\item For any $T \subseteq \widehat{V}$,
$\omega(T \cap V) = \omega(T)$;
moreover, if $w(\widehat{\partial}^+(T)) < \mathfrak{M}$,
then $w(\partial^+(T \cap V))=w(\widehat{\partial}^+(T))$.
\end{compactitem}

\end{fact}

Fact~\ref{fact:reduction} implies that for problems such
as directed sparsest cut (with product demands), max-flow  and min-cut,
it suffices to consider directed normal graphs.

\noindent \textbf{Semidefinite Program (SDP) Formulation.}
Arora et al.~\cite{jacm/AroraRV09} formulated an SDP
for the sparsest cut problem with uniform demands
for undirected normal graphs.
The SDP was later refined by Agarwal et al.~\cite{stoc/AgarwalCMM05} for directed normal graphs to give a rounding-based approximation algorithm.
Since the method can be easily generalized to product demands with $\kappa$-skewed vertex weights by
duplicating copies,
we have the following corollary.

%
%
%

\begin{corollary}[Approximation Algorithm for Directed Sparsest Cut with Product Demands]
\label{cor:main_expansion}
~~For the directed sparsest cut problem
with product demands (with $\kappa$-skewed vertex weights) on directed hypergraphs,
there are randomized polynomial-time $O(\sqrt{\log \kappa n})$-approximate algorithms.
\end{corollary}

\ignore{
\noindent \textbf{Approximation Algorithm for Directed Hyperedge Expansion.}
Similar to previous works~\cite{LeightonR99,jacm/AroraRV09},
we reduce the directed hyperedge expansion problem
to the directed sparsest cut problem with \emph{product demands}
in directed hypergraphs, preserving a multiplicative factor of 2 in the approximation ratio.
By product demands, we mean that there is some weight function $\omega: V \rightarrow \Z_+$
such that the directed sparsity of a subset $S$ is $\vartheta(S) := \frac{w(\partial^+(S))}{\omega(S) \cdot \omega(\overline{S})}$.
}



\noindent \emph{Are we done yet?} Unfortunately, solving an SDP poses a major
bottleneck in running time. Alternatively, Arora and Kale~\cite{AroraK16} proposed an SDP primal-dual framework that iteratively updates the primal and the dual solutions.

\noindent \emph{Outline of SDP Primal-Dual Approach.}
The framework essentially performs binary search on the optimal SDP value.
Each binary search step requires iterative calls to some \oracle.
Loosely speaking, given a (not necessarily feasible) primal solution candidate
of the minimization SDP, each call of the \oracle returns either (i) a subset $S \subset V$ of small enough sparsity $\vartheta(S)$, or (ii)
a (not necessarily feasible) dual solution with large enough objective value
to update the primal candidate in the next iteration.  At the end of the last iteration,
if a suitable subset $S$ has not been returned yet,
then the dual solutions returned in all the iterations
can be used to establish that the optimal SDP value is large.

\noindent \emph{Disadvantage of Direct Reduction.}
For directed sparsest cut problem with uniform demands,
the primal-dual framework gives an $O(\sqrt{\log n})$-approximate algorithm,
which
has running time\footnote{\label{ft:1}
After checking the calculation in~\cite{kale2007efficient} carefully,
we conclude that there should actually be an extra factor of $O(n^2)$ in the running time.
Through personal communication with Kale,
we are told that it might be possible reduce a factor of $O(n)$,
using the ``one-sided width'' technique in~\cite{kale2007efficient}.
}  $\widetilde{O}(m^{1.5} + n^{2 + o(1)})$.
If we apply the reduction in Fact~\ref{fact:reduction} directly,
the resulting running time for directed hypergraphs
becomes $\widetilde{O}((mr)^{1.5} + n^{2 + o(1)} + m^{2 + o(1)})$.
The term $(mr)^{1.5}$ is due to a max-flow computation, which
is not obvious how to improve.  However, the extra $m^{2 + o(1)}$ term is introduced,
because the dimension of the primal domain is increased.
Therefore, we think it is worthwhile to adapt
the framework in~\cite{kale2007efficient} to directed hypergraphs
to avoid the extra $m^{2 + o(1)}$ term.

\noindent \emph{Other Motivations.}
We deconstruct the algorithm for directed normal graphs with uniform vertex weight in
Kale's PhD thesis~\cite{kale2007efficient},
and simplify the notation.  The result
is a much cleaner description of the algorithm, even though we
consider more general directed hypergraphs and non-uniform vertex weights.
As a by-product, we discover that
since the subset returned by sparsest cut needs not be balanced,
there should be an extra factor of $O(n^2)$ in the running time of their algorithm.
We elaborate the details further as follows.


\begin{enumerate}
\item In their framework,
they assume that in the SDP, there is some constraint on the trace $\Tr(\bm{X}) = \bm{I} \bullet \bm{X}$, which can be viewed as some dot-product with the identity matrix~$\bm{I}$.  The important property is that every non-zero vector is an eigenvector of $\bm{I}$ with eigenvalue 1.  Therefore, if the smallest eigenvalue of $\bm{A}$ is at least $-\epsilon$
for some small $\epsilon > 0$, then the sum $\bm{A} + \epsilon \bm{I} \succeq 0$
has non-negative eigenvalues.  This is used crucially to establish a lower bound
on the optimal value of the SDP.

However, for the SDP formulation of directed sparsest cut,
the constraint loosely translates to $\bm{I} \cdot \bm{X} \leq O(\frac{n}{\omega(S) \cdot \omega(\overline{S})})$, where $S \subset V$ is some candidate subset.
To achieve the claimed running time, one needs a good enough upper bound,
which is achieved if the subset is balanced.  However, for general $S$ that is not balanced,
there can be an extra factor of $n$ in the upper bound, which translates to
a factor of $O(n^2)$ in the final running time.

Instead, as we shall see, there is already a constraint $\bm{K} \bullet \bm{X} = 1$,
where $\bm{K}$ is the Laplacian matrix of the complete graph.
Since $\bm{K}$ is actually a scaled version of the identity operator
on the space orthogonal to the all-ones vector $\one$,
a more careful analysis can use this constraint involving $\bm{K}$ instead.

\item  In capturing directed distance in an SDP~\cite{stoc/AgarwalCMM05}, typically, one extra vector $v_0$ is added. However, in the SDP of~\cite{kale2007efficient},
a different vector $w_i$ is added for each~$i \in V$, and
constraints saying that all these $w_i$'s are the same are added.
At first glance, these extra vectors $w_i$'s and constraints seem extraneous, and create a lot
of dual variables in the description of the \oracle.
The subtle reason is that by increasing the dimension of the primal domain,
the width of the \oracle, which is measured by the spectral norm of some matrix,
can be reduced.


Observe that the matrix $\bm{K}$ does not involve any extra added vectors.
If we do not use the trace bound on $\Tr(\bm{X})$ in the analysis,
then we cannot add any extra vectors in the SDP.
This can be easily rectified, because we can just label any vertex in $V$ as $0$
and consider two cases.  In the first case, we formulate an SDP for the solution $S$
to include $0$; in the second case, we formulate a similar SDP to exclude $0$ from the solution.  The drawback is that now the width of the \oracle increases by a factor of $O(n)$,
which leads to a factor of $O(n^2)$ in the number of iterations.

Therefore, in the end, we give a simpler presentation than~\cite{kale2007efficient},
but the asymptotic running time is the same,
although an improvement as mentioned in Footnote~\ref{ft:1} might be possible.


\item For each simple path, they add a generalized $\ell_2^2$-triangle inequality. This causes
an exponential number of dual variables (even though most of them are zero).
However, only triangle inequalities for triples are needed, because
each triangle inequality for a long path is
just a linear combination of inequalities involving only triples.

\ignore{
On the other hand, we would like to point out that our triangle inequalities
do not involve the extra vector~$v_0$.  The high-level reason
is that the Laplacian matrix $\bm{K}$ of the complete graph on $V$
(described in the above analysis) does not involve $v_0$.}
\end{enumerate}

We summarize the performance of our modified primal-dual approach as follows.

\begin{theorem}[SDP Primal-Dual Approximation Algorithm
for Directed Sparsest Cut]
Suppose the vertex weights are $\kappa$-skewed.
Each binary search step of the primal-dual framework takes
$T := \widetilde{O}(\kappa^2 n^2)$ iterations.
The running time of each iteration is
$\widetilde{O}((rm)^{1.5} + (\kappa n)^{2})$.

The resulting approximation ratio is $O(\sqrt{\log \kappa n})$.

%
\end{theorem}

%
%
%
%

\subsection{Related Work}

As mentioned above, the most related work is the SDP primal-dual framework
by Arora and Kale~\cite{AroraK16} used for solving various variants
of the sparsest cut problems.  The details for directed sparsest cut
are given in Kale's PhD thesis~\cite{kale2007efficient}.  We briefly describe
the background of related problems as follows.

\noindent \emph{Edge Expansion and Sparsest Cut.}  Leighton and Rao~\cite{LeightonR99}
achieved the first $O(\log n)$-approximation algorithms for the edge expansion problem
and the sparsest cut problem with general demands for undirected normal graphs.
An SDP approach utilizing $\ell_2^2$-representation was used by Arora et al.~\cite{jacm/AroraRV09}
to achieve $O(\sqrt{\log n})$-approximation for the special case of uniform demands;
subsequently, $O(\sqrt{\log n} \cdot \log \log n)$-approximation has been achieved
for general demands~\cite{arora2008euclidean} via embeddings of $n$-point $\ell_2^2$ metric spaces into
Euclidean space with distortion $O(\sqrt{\log n} \cdot \log \log n)$.
This embedding was also used
to achieve $O(\sqrt{\log n \log r} \cdot \log\log n)$-approximation for
the general demands case in undirected hypergraphs~\cite{chan2018jacm},
where $r$ is the maximum cardinality of an hyperedge.

For directed graphs, Agarwal et al.~\cite{stoc/AgarwalCMM05}
generalized the separator theorem~\cite{jacm/AroraRV09} for $\ell_2^2$-representation vectors
to the directed case and achieved an $O(\sqrt{\log n})$-approximation
for the directed sparsest cut problem with uniform demands.  The general demands variant
for directed graphs seems to be much harder, as the best currently known polynomial-time
approximation ratio is $\tilde{O}(n^{\frac{11}{23}})$ by~\cite{stoc/AgarwalAC07}.

An $O(\sqrt{\log n})$-approximation for undirected hyperedge expansion
has been achieved by Louis and Makarychev~\cite{louis2014approximation},
who used hypergraph orthogonal separator as the main tool in rounding their SDP formulation.
However, their orthogonal separator technique is more suitable for dealing with undirected hypergraphs.
It is not immediately clear how to generalize their orthogonal separator to directed hypergraphs.
Instead, we follow the approach in~\cite{stoc/AgarwalCMM05} and still can achieve
the same approximation ratio of $O(\sqrt{\log n})$ for directed hyperedge expansion.

\section{SDP Relaxation for Directed Sparsest Cut}
\label{sec:sdp}



\ignore{
Hence, for the rest of the section, we focus on the latter problem.

Observe that product demands are more general than uniform demands,
but are still special cases of general demands.}


We follow some common notation concerning
sparsest cut (with uniform demands) in undirected~\xcite{jacm/AroraRV09}
and directed~\xcite{stoc/AgarwalCMM05} normal graphs.

\begin{definition}[$\ell_2^2$-Representation]
	An $\ell_2^2$-representation for a  set of vertices $V$  is an assignment of a vector $v_i$ to each vertex $i\in V$ such that the $\ell_2^2$-triangle inequality holds:
	\begin{equation*}
	\norm{v_i-v_j}^2 \leq \norm{v_i-v_k}^2 +\norm{v_k-v_j}^2,\qquad \forall i,j,k\in V.
	\end{equation*}
\end{definition}

\noindent \textbf{Directed Distance~\xcite{stoc/AgarwalCMM05}.}
We arbitrarily pick some vertex in $V$, and call it $0$.

We first consider the case when $0$ is always included in the feasible solution.
Given an $\ell_2^2$-representation $\{v_i\}_{i \in V}$,
define the directed distance $d: V \times V \rightarrow \R_+$
by 

$d(i,j) := \norm{v_i - v_j}^2 - \norm{v_i - v_0}^2 + \norm{v_j - v_0}^2$.  

It is easy to verify the directed triangle inequality:
for all $i,j,k\in V$,
$d(i,k)+d(k,j) 	\geq d(i,j)$.

\ignore{
It is easy to verify that $d$ satisfies the following.
\begin{compactitem}
\item For all $i,j\in V$, $d(i,j)\geq 0$; for all $i\in V$, $d(i,i) = 0$.
\item Directed triangle inequality:
for all $i,j,k\in V$,
$d(i,k)+d(k,j) 	\geq d(i,j)$.
\end{compactitem}
}

\noindent For subsets $S \subseteq V, T \subseteq V$, we also denote
$d(S,T) := \min_{i\in S,j\in T}\{d(i,j)\}$,
$d(i, S) = d(\{i\}, S)$ and $d(S,i) = d(S,\{i\})$.

\noindent \emph{Interpretation.} In an SDP-relaxation
for directed sparsest cut, 
vertex~$0$ is always chosen in the solution~$S \subseteq V$.
For $i \in S$, $v_i$ is set to $v_0$;
for $i \in \overline{S} = V \setminus S$, $v_i$ is set to $-v_0$.
Then, it can be checked that $d(i,j)$ is non-zero
\emph{iff} $i \in S$ and $j \in \overline{S}$, in which case
$d(i,j) = 8 \|v_0\|^2$.

\noindent \emph{The other case.} For the other case when $0$ is definitely excluded
from the solution~$S$,
it suffices to change the definition
$d(i,j) := \norm{v_i - v_j}^2 - \norm{v_i + v_0}^2 + \norm{v_j + v_0}^2$.
For the rest of the paper, we just concentrate on the case that $0$ is in the solution $S$.


We consider the following SDP relaxation (where $\{v_i: i \in V\}$
are vectors) for the directed sparsest cut
problem with product demands on an edge-weighted
hypergraph $H=(V, E, w)$ with
vertex weights $\omega: V \rightarrow \{1, 2, \ldots, \kappa\}$.
We denote $\ws := \sum_{i \in V} \omega_i$.

\begin{align}
\SDP \qquad \min \qquad &\frac{1}{2}\sum_{e\in E}w_e\cdot d_e  \label{eq:SDP}\\
\text{s.t.} \qquad &d_e \geq  d(i,j), \quad \quad \quad
  \forall e\in E, \forall (i,j)\in\Te\times\He \\
&\norm{v_i - v_j}^2 \leq \norm{v_i-v_k}^2 + \norm{v_k-v_j}^2,  &\forall i,j,k\in V 
\label{constraint:triangle} \\
&\sum_{\{i,j\}\in{V \choose 2}} \omega_i \omega_j \norm{v_i-v_j}^2 =1, \label{constraint:sum_one} \\
& d_e \geq 0, \quad \forall e \in E.
\end{align}

\ignore{
\noindent Since this is a minimization problem, without loss of generality,
we only need to consider feasible solutions that satisfy
$d_e = \max_{(i,j)\in\Te\times \He}\{ d(i,j) \}$ for all $e\in E$.
}

\noindent \emph{SDP Relaxation.}
To see that $\SDP$ is a relaxation of the directed sparsest cut problem, it suffices to show that
any subset $S \subseteq V$ induces a feasible solution with objective function $\vartheta(S)$.
We set $v_0$ to be a vector with $\norm{v_0}^2 = \frac{1}{4\omega(S)\cdot \omega(\overline{S})}$.
For each $i \in V$, we set $v_i := v_0$ if $i \in S$, and $v_i := - v_0$ if $i \in \overline{S}$.
Then, the value of the corresponding objective is
\begin{align*}
&\frac{1}{2}\sum_{e\in E}w_e\cdot d_e = \frac{1}{2}\sum_{e\in E}w_e\cdot \max_{(i,j)\in\Te\times \He}\{ d(i,j) \}\\
= &\frac{1}{2}\sum_{e\in \partial^+(S)} w_e\cdot(\norm{v_0+v_0}^2 - \norm{v_0-v_0}^2 + \norm{-v_0-v_0}^2) = \frac{w(\partial^+(S))}{\omega(S)\cdot \omega(\overline{S})} = \vartheta(S).
\end{align*}

\noindent \emph{Trace Bound.}  
We have $\sum_{i \in V} \norm{v_i}^2 \leq \frac{n}{4 \omega(S) \cdot \omega(\overline{S})} \leq O(\frac{\kappa n^2}{\ws^2})$.
Note that if $S$ is balanced, then
the upper bound can be improved to $O(\frac{n}{\ws^2})$.


\noindent \textbf{SDP Primal-Dual Approach~\cite{AroraK16}.}  Instead of solving the SDP directly, the SDP is used as a tool for finding
an approximate solution.  Given a candidate value~$\alpha$, the primal-dual approach either
(i) finds a subset $S$ such that $\vartheta(S) \leq O(\sqrt{\log n}) \cdot \alpha$,
or (ii) concludes that the optimal value of the SDP is at least~$\frac{\alpha}{2}$.
Hence, binary search can be used
to find an $O(\sqrt{\log n})$-approximate solution.  This approach is described in Section~\ref{sec:primaldual}.

\ignore{
\noindent \textbf{Two SDP Approaches.} We give the following methods to achieve an $O(\sqrt{\log n})$-approximate solution.
\begin{compactitem}
\item Rounding Approach. One can solve the SDP~(\ref{eq:SDP}) and round the solution to give some subset~$S$.
This is described in Appendix~\ref{sec:rounding}.
\item 
\end{compactitem}
}

\section{SDP Primal-Dual Approximation Framework}
\label{sec:primaldual}

We use the primal-dual framework by~\cite{AroraK16}.
However, instead of using it just as a blackbox,
we tailor it specifically for our problem to have a cleaner description.

\noindent \textbf{Notation.}
We use a bold capital letter $\bm{A} \in \R^{V \times V}$
to denote a symmetric matrix whose rows and columns
are indexed by $V$.

The sum of the diagonal entries of a square matrix $\bm{A}$ is
denoted by the trace $\Tr(\bm{A})$.  Given two matrices $\bm{A}$ and $\bm{B}$,
let $\bm{A} \bullet \bm{B} := \Tr(\bm{A}^\T \bm{B})$, where $\bm{A}^\T$ is the transpose of $\bm{A}$.
We use $\one \in \R^{V}$ to denote the all-ones vector.



\noindent \emph{Primal Solution.} We use $\bm{X} \succeq 0$ to denote a positive semi-definite matrix
that is associated with the vectors $\{v_i\}_{i \in V}$
such that $\bm{X}(i,j) = \langle v_i, v_j \rangle$.

We rewrite \SDP~(\ref{eq:SDP}) to an equivalent form as follows.

\begin{align}
\SDP \qquad \min \qquad & \frac{1}{2}\sum_{e \in E} w_e \cdot d_e &  &\label{eq:SDP2}\\
\text{s.t.} \qquad      & d_e - \bm{A}_{ij} \bullet \bm{X}  \geq \,  0,  \quad \quad
 \forall e \in E, \forall (i,j)\in \Te\times\He \\
& \bm{T}_p \bullet \bm{X}  \geq \,  0,  \quad \quad \quad \quad \quad
 \forall p \in \mathcal{T} \label{constraint:triangle} \\
& \bm{K} \bullet \bm{X}  = \,  1, & \label{constraint:sum_one} \\
& \bm{X}  \succeq 0\,; \quad \quad d_e  \geq \,  0, \quad \forall e \in E.
\end{align}

We define the notation used in the above formulation as follows:
\begin{compactitem}
\item For $(i,j) \in V \times V$,
$\bm{A}_{ij}$ is the unique symmetric matrix such that $\bm{A}_{ij} \bullet \bm{X} = d(i,j) = \norm{v_i - v_j}^2 - \norm{v_i - v_0}^2 + \norm{v_j - v_0}^2$.

Since we consider a minimization problem,
we just use $\bm{X} \succeq 0$ to represent a primal solution,
and automatically set $d_e := \max\{ 0, \max_{(i,j)\in \Te\times\He} \bm{A}_{ij} \bullet \bm{X}\}$ for all $e \in E$.  As we shall see, this implies that
corresponding dual variable $y^e_{ij} \in \R$ can be set to 0.

Moreover, we do not need the constraint $\bm{A}_{ij} \bullet \bm{X} \geq 0$,
because we already have $d_e \geq 0$.

\item The set $\mathcal{T}$  contains elements of the form
$\tri{i}{j}{k} \in {V \choose 2} \times V$,
where $i, j, k$ are distinct elements in $V$.


They are used to specify
the $\ell_2^2$-triangle inequality.

For $p = \tri{i}{j}{k}$,
$\bm{T}_p$ is defined such that $\bm{T}_p \bullet \bm{X} = \norm{v_i - v_j}^2 + \norm{v_j - v_k}^2 - \norm{v_i - v_k}^2$.

Observe that in~\cite{kale2007efficient},
a constraint is added for every path in the complete graph on $V$.
However, these extra constraints are simply linear combinations
of the triangle inequalities, and so, are actually unnecessary.

\item As above, $\bm{K}$ is defined such that $\bm{K} \bullet \bm{X} =
\sum_{\{i, j\} \in \binom{V}{2}} \omega_i\omega_j\norm{v_i - v_j}^2$.

Observe that any $\bm{X} \succeq 0$ can be re-scaled
such that $\bm{K} \bullet \bm{X} = 1$.

\item \textbf{Optional constraint.}
%
In~\cite{kale2007efficient},
an additional constraint is added, which in our notation\footnote{
In the original notation~\cite[p.59]{kale2007efficient},
the claimed constraint is $\Tr(\bm{X}) \leq n$,
but for general cut~$S$, only the weaker bound $\Tr(\bm{X}) \leq \Theta(n^2)$ holds.
} becomes:

$$ - \bm{I} \bullet \bm{X} \geq - \Theta(\frac{n}{\ws^2}).$$

However, this holds only if the solution $S$ is balanced.  For general cut~$S$,
we only have the weaker bound: $ - \bm{I} \bullet \bm{X} \geq - \Theta(\frac{\kappa n^2}{\ws^2}).$
As we shall see in the proof of Lemma~\ref{lemma:correct},
adding this weaker bound is less useful than the above constraint
$\bm{K} \bullet \bm{X} = 1$.

\end{compactitem}

\ignore{
where $X \in \R^{(n + 1) \times (n + 1)}$,
$A_{ij}$, $T_{j, \{i, k\}}$ and $K$ are symmetric matrices such that
$A_{ij} \bullet X = d(i, j)$,
$T_{j, \{i, k\}} \bullet X = \norm{v_i - v_j}^2 + \norm{v_j - v_k}^2 - \norm{v_i - v_k}^2$ and
$K \bullet X = \sum_{\{i, j\} \in \binom{V}{2}} \omega_i\omega_j\norm{v_i - v_j}^2$ when $X = V^\transpose V$,
where $V$ has column vectors $v_0, v_1, \dots, v_n$.

We denote $W := \sum_{i = 1}^n \omega_i$ and $\delta := \min_{i \in [n]}\frac{n\omega_i}{W}$.
}

The dual to $\SDP$ is as follows:

\begin{align}
\Dual \qquad \max \qquad & z \\
\text{s.t.} \qquad &-\sum_{e \in E}\sum_{(i, j) \in \Te \times \He} y^e_{ij} \bm{A}_{ij} +
\sum_{p \in \mathcal{T}}  f_p \bm{T}_p +  z \bm{K}  \, \preceq  \, 0 \\
& \sum_{(i, j) \in \Te \times \He}y^e_{ij} \quad  \le  \, \frac{w_e}{2},  \quad \forall e \in E, \\
& f_p  \ge  0, \quad \forall p \in \mathcal{T}, \\
& y^e_{ij}  \ge 0, \quad \forall e \in E, \forall (i, j) \in \Te \times \He.
\end{align}

Observe that,
if we add the optional constraint $- \bm{I} \bullet \bm{X} \geq - b$ in the primal,
then this will create a dual variable $x \geq 0$,
which causes an extra term $- b x$ in the objective function
and an extra term $- x \bm{I}$ on the left hand side of the constraint.
\ignore{
However, to have a unified description,
we will always set this variable $x$ to 0, and hence omit it in the description.
}

To use the primal-dual framework~\cite{AroraK16},
we give a tailor-made version of the \oracle for our problem.

\begin{definition}[\oracle for \SDP]
\label{defn:oracle}
Given $\alpha > 0$, $\oracle(\alpha)$ has width $\rho$ (which can depend on $\alpha$) if the following holds.
Given a primal candidate solution $\bm{X} \succeq 0$ (associated with
vectors $\{v_i\}_{i \in V}$) such that $\bm{K} \bullet \bm{X} = 1$,
it outputs either

(i) a subset $S \subsetneq V$ such that its sparsity
$\vartheta(S) \leq O(\sqrt{\log \kappa n}) \cdot \alpha$, or

(ii) some dual variables $(z, (f_p \geq 0: p \in \mathcal{T}))$, where all $y^e_{ij}$'s are implicitly 0, and a symmetric \emph{flow} matrix  $\bm{F} \in \R^{V \times V}$ such that all
the following hold:

\begin{compactitem}
\item $z \geq \alpha$

\item $(\sum_{p \in \mathcal{T}} f_p \bm{T}_p + z \bm{K} ) \bullet \bm{X} \le \bm{F} \bullet \bm{X}$

\item For all \emph{feasible} primal solution $\bm{X}^*$,
$\bm{F} \bullet \bm{X}^* \leq \frac{1}{2} \sum_{e \in E}  w_e d^*_e$,

where $d^*_e := \max\{0, \max_{(i, j) \in \Te \times \He} \bm{A}_{ij} \bullet \bm{X}^*\}$.

\item For all $x \in \spn\{\one\}$, $\bm{F} x = 0$.

\item The spectral norm $\norm{\sum_{p \in \mathcal{T}} f_p \bm{T}_p + z \bm{K} - \bm{F}}$ is at most $\rho$.
\end{compactitem}

\ignore{
in which $d_e = \max_{(i, j) \in \Te \times \He}A_{ij} \bullet X$ for all $e \in E$,
either outputs dual variables $(f, z)$ and a matrix $F$ such that
\begin{align}
z \ge \alpha \\
\left(\sum_p f_p T_p + zK\right) \bullet X \le F \bullet X \\
f_p \ge 0, \qquad \forall p \\
F \bullet X \le \frac{1}{2}\sum_{e \in E}w_e d_e \qquad \text{for all feasible } X \\
F \cdot 1 = 0 \\
\norm{\sum_p f_p T_p + zK - F}^2 \le \rho
\end{align}
or outputs a vertex set $S$ with $\vartheta(S) = O(\sqrt{\log n} \cdot \alpha)$.
}
\end{definition}

Using \oracle in Definition~\ref{defn:oracle},
we give the primal-dual framework for one step of the binary search
in Algorithm~\ref{alg:primaldual}.
As in~\cite{AroraK16}, the running for each iteration is dominated
by the call to the \oracle.

\begin{algorithm}[H]
\caption{Primal-Dual Approximation Algorithm for \SDP}
\label{alg:primaldual}
\KwIn{Candidate value $\alpha > 0$; $\oracle(\alpha)$ with width $\rho$}
$T$ is chosen as in Lemma~\ref{lemma:correct};
$\eta \leftarrow \sqrt{\frac{\ln n}{T}}$\;
$\bm{W}^{(1)} \leftarrow I \in \R^{V \times V}$\;
\For{$t = 1, 2, \dots, T$}{
    $\bm{X}^{(t)} \leftarrow \frac{\bm{W}^{(t)}}{\bm{K} \bullet \bm{W}^{(t)}}$\;
    Run $\oracle(\alpha)$ with $X^{(t)}$\;
    \If{\oracle returns some $S \subset V$}{
        \textbf{return} $S$ and \textbf{terminate}.
    }
		
		Otherwise, the \oracle returns some dual solution $(z^{(t)}, (f_p^{(t)}: p \in \mathcal{T}))$
		and matrix $\bm{F}^{(t)}$
		as promised in Definition~\ref{defn:oracle}.

    $\bm{M}^{(t)} \leftarrow -\frac{1}{\rho}\left(\sum_{p \in \mathcal{T}} f_p^{(t)} \bm{T}_p
		+ z^{(t)} \bm{K} - \bm{F}^{(t)}\right)$\;
    $\bm{W}^{(t + 1)} \leftarrow \exp\left(-\eta\sum_{\tau = 1}^t \bm{M}^{(\tau)}\right)$\;
}
\If{no subset $S$ is returned yet}{
        \textbf{report} the optimal value is at least $\frac{\alpha}{2}$.
    }
\end{algorithm}

The following result is proved in~\cite[Corollary 3.2]{AroraK16}

\begin{fact}[Multiplicative Update]
\label{fact:mwalg}
Given any sequence of matrices $\bm{M}^{(1)}, \bm{M}^{(2)}, \dots, \bm{M}^{(T)} \in \R^{n \times n}$
that all have spectral norm at most~$1$ and $\eta \in (0, 1]$,
let $\bm{W}^{(1)} = I$, $\bm{W}^{(t)} = \exp\left(-\eta\sum_{\tau = 1}^{t-1} \bm{M}^{(\tau)}\right)$, for $t = 2, \dots, T$; let
$\bm{P}^{(t)} = \frac{\bm{W}^{(t)}}{\Tr(\bm{W}^{(t)})}$, for $t = 1, 2, \dots, T$.
Then, we have

$$\sum_{t = 1}^T \bm{M}^{(t)} \bullet \bm{P}^{(t)}
\le \lambda_{\mathsf{min}}\left(\sum_{t = 1}^T \bm{M}^{(t)}\right)
+ \eta T + \frac{\ln n}{\eta},$$

\noindent where $\lambda_{\mathsf{min}}(\cdot)$ gives the minimum eigenvalue
of a symmetric matrix.
\end{fact}

\begin{lemma}[Correctness]
\label{lemma:correct}
Set $T := \ceil{\frac{16 \kappa^2 \rho^2 n^2 \ln n}{\alpha^2 \ws^4}}$.
Suppose that in Algorithm~\ref{alg:primaldual},
the \oracle never returns any subset $S$ in any of the $T$ iterations.
Then, the optimal value of \SDP is at least $\frac{\alpha}{2}$.

\ignore{
if we set the value of $T$ in each of the following cases.

\begin{compactitem}
\item[(a)] For the unrestricted version, set $T := \ceil{\frac{16 \kappa^2 \rho^2 n^2 \ln n}{\alpha^2 \ws^4}}$.

\item[(b)] For the balanced version,
one has the extra constraint $ - \bm{I} \bullet \bm{X} \geq - \Theta(\frac{n}{\ws^2})$;
we set $T := \Theta(\frac{\rho^2 n^2 \ln n}{\alpha^2 \ws^4})$ in this case.

\end{compactitem}
}

\end{lemma}

\begin{proof}
The proof follows the same outline as~\cite[Theorem 4.6]{AroraK16},
but we need to be more careful,
depending on whether we use the constraint on $\bm{I} \bullet \bm{X}$.

\ignore{
because
we only have the primal constraint $\bm{K} \bullet \bm{X} = 1$,
as opposed to some bound on $\Tr(\bm{X})$.
}

For $t = 1, \ldots, T$, we use $\bm{M}^{(t)}$ as in Algorithm~\ref{alg:primaldual},
and apply Fact~\ref{fact:mwalg}.
Definition~\ref{defn:oracle} guarantees that
$\bm{M}^{(t)} \bullet \bm{P}^{(t)} \ge 0$,
because $\bm{X}^{(t)}$ is positively scaled from $\bm{P}^{(t)}$.

Hence, by Fact~\ref{fact:mwalg},
we have
$\lambda_{\mathsf{min}}\left(\sum_{t = 1}^T \bm{M}^{(t)}\right) + \eta T + \frac{\ln n}{\eta} \ge 0.$

By setting $\eta := \sqrt{\frac{\ln n}{T}}$
and $\bm{Z} := \frac{\rho}{T}\sum_{t = 1}^T \bm{M}^{(t)} =
\frac{1}{T} \sum_{t=1}^T(\bm{F}^{(t)} - \sum_{p \in \mathcal{T}} f_p^{(t)} \bm{T}_p
		- z^{(t)} \bm{K})$,
this is equivalent to $\lambda_{\mathsf{min}}(\bm{Z}) \geq -2\rho \cdot \sqrt{\frac{\ln n}{T}}$.

As in~\cite{AroraK16}, we would like to add some matrix from the primal constraint
to $\bm{Z}$ to make the resulting matrix positive semi-definite.

A possible candidate is $\bm{K}$, whose eigenvalues are analyzed as follows.

First, observe that for all $x \in \spn\{\one\}$,
it can be checked that $\bm{K} x = \bm{T}_p \, x = 0$, for all $p \in \mathcal{T}$.
Furthermore, Definition~\ref{defn:oracle}
guarantees that $\bm{F}^{(t)} x = 0$, for all $t$.
Hence, it follows that $\bm{Z} x = 0$,
which implies that any negative eigenvalue of $\bm{Z}$ must be due to
the space orthogonal to $\spn\{\one\}$.

\ignore{
Observe that it is crucial here that $V = [V]$, i.e., one of the
vertices in $V$ is labeled as 0.  Otherwise, the $0$-eigenspace of $\bm{K}$
has rank at least 2, but we can only be sure that $\bm{F}^{(t)}$
has at least 1 eigenvector with eigenvalue $0$.
}

We next analyze the eigenvectors of $\bm{K}$ in this orthogonal space.
Consider a unit vector $u \perp \spn\{\one\}$, i.e.,
$\sum_{i \in V} u_i = 0$ and $\sum_{i \in V} u_i^2 = 1$.

Then, $u^\T K u = \frac{1}{2} \sum_{i \in V} \sum_{j \in V} \omega_i \omega_j (u_i - u_j)^2
= \ws^2 \cdot [\sum_{i \in V} \delta_i u_i^2 -  (\sum_{i \in V} \delta_i u_i)^2]$,
where $\delta_i := \frac{\omega_i}{\ws}$ can be interpreted
as some probability mass function.
Hence, this term can be interpreted as some variance.

Observe that the $\kappa$-skewness of the weights $\omega$
implies that for all $i \in V$, $\delta_i \geq \frac{1}{\kappa n}$.
Therefore, Lemma~\ref{lem:var} below implies that
$u^\T K u \geq  \ws^2 \cdot \frac{1}{\kappa n}$.

Hence, by enforcing
$\epsilon \cdot \ws^2 \cdot \frac{1}{\kappa n} \ge 2 \rho \cdot \sqrt{\frac{\ln n}{T}}$,
we have $\lambda_{\mathsf{min}}(\bm{Z} + \epsilon \bm{K}) \geq 0$.

Next, suppose $\bm{X}^*$ (with induced $d^*$) is an optimal primal solution to \SDP.
Then, Definition~\ref{defn:oracle} implies that
$\frac{1}{2} \sum_{e \in E} w_e d^*_e \geq \frac{1}{T} \sum_{t=1}^T \bm{F}^{(t)} \bullet \bm{X}^*$.

Since $(\bm{Z} + \epsilon \bm{K}) \bullet \bm{X}^* \geq 0$, the optimal value
is at least

\begin{align}
 &  \frac{1}{T}\sum_{t = 1}^T \left(\sum_{p \in \mathcal{T}} f_p^{(t)} \bm{T}_p \bullet \bm{X}^*\right)
+ \frac{1}{T}\sum_{t = 1}^T z^{(t)} \bm{K} \bullet \bm{X}^* - \epsilon \bm{K} \bullet \bm{X}^* \\
 \ge & 0 + \frac{1}{T}\sum_{t = 1}^T z^{(t)} \cdot 1 - \epsilon \cdot 1 \\
 \ge & \alpha - \epsilon,
\end{align}

\noindent where the last two inequalities come from the properties of primal feasible $\bm{X}^*$ and \oracle, respectively. Setting $\epsilon = \frac{\alpha}{2}$ gives the result.
\qed
\end{proof}

\noindent \textbf{Remark.} One can see that in the proof of Lemma~\ref{lemma:correct},
if one uses the weaker bound  $ - \bm{I} \bullet \bm{X} \geq - \Theta(\frac{\kappa n^2}{\ws^2})$.
Then, the proof continues by choosing
$\nu  = 2 \rho \cdot \sqrt{\frac{\ln n}{T}}$,
we have $\lambda_{\mathsf{min}}(\bm{Z} + \nu \bm{I}) \geq 0$.

Using the same argument,
we conclude that the optimal value is at least

$\alpha - \nu \bm{I} \bullet \bm{X}^* \geq \alpha -  \nu \cdot \Theta(\frac{\kappa n^2}{\ws^2})$.

Setting $\frac{\alpha}{2} = \nu \cdot \Theta(\frac{\kappa n^2}{\ws^2})$
gives
$T := \Theta(\frac{\kappa^2 \rho^2 n^4 \ln n}{\alpha^2 \ws^4})$ in this case,
which has an extra factor of $O(n^2)$.

However, since we do not add any extra vectors in our primal domain,
the width in our \oracle in Theorem~\ref{th:oracle} has an extra $O(n)$ factor compared to
that in~\cite{kale2007efficient}, which brings back the $O(n^2)$ factor we have saved earlier.

%
%
%
%
%
%
%

\begin{lemma}[Bounding the Variance]
\label{lem:var}
For real numbers $u_1, u_2, \dots, u_n$ and $\delta_0, \delta_1, \delta_2, \dots, \delta_n$ such that\\
$\sum_{i = 1}^n u_i = 0$,
$\sum_{i = 1}^n u_i^2 = 1$,
$\sum_{i = 1}^n \delta_i = 1$ and $\delta_i \ge \delta_0 > 0$, $\forall i$,
we have $\sum_{i = 1}^n \delta_i u_i^2 - \left(\sum_{i = 1}^n \delta_i u_i\right)^2 \ge \delta_0$.
Moreover, we have $\sum_{i = 1}^n \delta_i u_i^2 - \left(\sum_{i = 1}^n \delta_i u_i\right)^2 \leq \max_{i} \delta_i$.

\end{lemma}
\begin{proof}
Let $u_1, \dots, u_n$ be fixed and consider the function
$$g(\delta_1, \dots, \delta_n) = \sum_{i = 1}^n \delta_i u_i^2 - \left(\sum_{i = 1}^n \delta_i u_i\right)^2$$
with domain $\{(\delta_1, \dots, \delta_n)|\sum_{i = 1}^n \delta_i = 1, \delta_i \ge \delta_0, \forall i\}$.

We claim that the minimum can be obtained at some point where at most one $\delta_i$
has value strictly greater than $\delta$.
Indeed, suppose there are two variables, say $\delta_1$ and $\delta_2$,
whose value is strictly greater than $\delta_0$.
Consider $h(x) = g(x, s - x, \delta_3, \dots, \delta_n)$
where $s = \delta_1 + \delta_2$.
Simplifying it, we know that the coefficient associated to $x^2$ in $h$ is $-(u_1 - u_2)^2 \le 0$,
which means that we can shift either $\delta_1$ or $\delta_2$ to $\delta_0$
(and the other variable to $s - \delta_0$) without increasing the value of $g$.

Therefore, we only need to consider the case where there is at most one $\delta_i > \delta_0$.
Without loss of generality, suppose $\delta_2 = \delta_3 = \dots = \delta_n = \delta_0$
and thus $\delta_1 = (1 - n \delta_0) + \delta_0$.
Then,
$$g(\delta_1, \dots, \delta_n) = \delta_0 + (1 - n \delta_0) u_1^2 - ((1 - n \delta_0)u_1)^2
= \delta_0 + n\delta_0(1 - n\delta_0)u_1^2 \ge \delta_0,$$
since $\delta_0$ cannot be greater than $\frac{1}{n}$.
\qed \end{proof}

\begin{corollary}[Non-zero Eigenvalues of $\bm{K}$]
\label{cor:K}
All eigenvectors of $\bm{K}$ that are orthogonal to $\one$
has eigenvalues in the range
$[\frac{\ws^2}{\kappa n}, \frac{\kappa \ws^2}{n}]$.
\end{corollary}

\section{Hypergraph Flows and Demands}
\label{sec:flow}

In order to facilitate the description of the \oracle,
we define some notation for flows and demands in hypergraphs.

\begin{definition}[Hypergraph Flow]
\label{defn:flow}
Given a directed hypergraph $H=(V,E)$,
a flow is defined as

$\mf{f} := (\mf{f}^e_{ij} \geq 0: e \in E, (i,j) \in \Te \times \He)$.

The corresponding flow matrix is defined as

$\bm{F} := \sum_{e \in E} \sum_{(i,j) \in \Te \times \He} \mf{f}^e_{ij} \bm{A}_{ij}$,
where $\bm{A}_{ij}$ is defined in Section~\ref{sec:primaldual}.

For $i \in V$,
the net amount of flow entering $i$
is

$\sum_{e \in E} ( \sum_{(j',i) \in \Te \times \He} \mf{f}^e_{j'i} - \sum_{(i,j) \in \Te \times \He} \mf{f}^e_{ij})$.

A flow $\mf{f}$ satisfies edge capacities $c : E \ra \R_+$
if for all $e \in E$, 

$\sum_{(i,j) \in \Te \times \He} \mf{f}^e_{ij} \leq c_e$.
\end{definition}

\begin{fact}[Constrained Flow]
\label{fact:cflow}
Suppose a flow $\mf{f}$ satisfies edge capacities $c$.
Then, for any primal solution $\bm{X} \succeq 0$ (with induced $\{d_e \geq 0 : e \in E\}$
as in defined is Section~\ref{sec:primaldual}), we have
$\bm{F} \bullet \bm{X} \leq \sum_{e \in E} c_e d_e$.
\end{fact}

Next, we define the notion of \emph{demand}, which is used to express
the sources and the sinks of a flow later.

\begin{definition}[Demand]
\label{defn:demand}
Given a vertex set $V$,
a demand between ordered pairs in $V$ is defined as

$\mf{d} := (\mf{d}_{ij} \geq 0: (i,j) \in V \times V)$.

The corresponding demand matrix is

$\bm{D} := \sum_{(i,j) \in V \times V} \mf{d}_{ij} \bm{A}_{ij}$.
\end{definition}

\begin{fact}[Spectral Norm of Demand Matrix~\cite{AroraK16}]
\label{fact:dem_norm}
The matrix for demand $\mf{d}$
has spectral norm

$\norm{\bm{D}} \leq O(\sum_{ij} \mf{d}_{ij})$.
\end{fact}

\begin{definition}[Flow Decomposition]
\label{defn:flow_decomp}
Suppose $p = (i_0, i_1, \ldots, i_k)$ is a directed path of length~$k \geq 2$.
Consider a flow of magnitude $f \geq 0$ along $p$,
which has the flow matrix $\bm{F}_p = \sum_{j=1}^k f_p \, \bm{A}_{i_{j-1}, i_j}$.

Then, the matrix can be expressed as

$\bm{F}_p = \sum_{j=1}^{k-1} f_p \, \bm{T}_{p_j}  + f_p \bm{A}_{i_0, i_k}$,

where $p_j := \tri{i_0}{i_j}{i_{j+1}} \in \mathcal{T}$,
and the notation is defined in Section~\ref{sec:primaldual}.

In general, any flow matrix can be decomposed as:

$\bm{F} = \sum_{p \in \mathcal{T}} f_p \, \bm{T}_p + \bm{D}$,
for some appropriate $f_p$'s and demand matrix $\bm{D}$.

Observe that the flow decomposition is not unique.
\end{definition}

\section{Implementation of $\oracle(\alpha)$}
\label{sec:oracle}

We give the implementation of the \oracle
as in Definition~\ref{defn:oracle}.
For $\alpha > 0$,
the input to the $\oracle(\alpha)$
is some $\bm{X} \succeq 0$
such that $\bm{K} \bullet \bm{X} = 1$.
By the standard Cholesky factorization,
we also have the associated vectors $(v_i : i \in V)$.
Then, we have:

$\sum_{\{i, j\} \in \binom{V}{2}} \omega_i\omega_j\|v_i - v_j\|^2 = 1.$

Below is the main result of this section.

\begin{theorem}
\label{th:oracle}
Given a candidate value $\alpha > 0$ and primal $\bm{X} \succeq 0$
such that $\bm{K} \bullet \bm{X} = 1$,
the \oracle ruturns one of the following:
\begin{enumerate}

		\item A subset $S$ with directed sparsity
        $\vartheta(S) = \frac{w(\partial^+(S))}{\omega(S)\omega(\overline{S})}
        = O(\sqrt{\log \kappa n}) \cdot \alpha$.

		\item Dual variables $(z, (f_p: p \in \mathcal{T}))$ and flow matrix $\bm{F}$
		satisfying Definition~\ref{defn:oracle}.
		
		Moreover, the spectral norm satisfies $\norm{\sum_p f_p \, \bm{T}_p + z \, \bm{K} - \bm{F}} \leq O(\alpha \ws^2 \sqrt{\log \kappa n})$.
		
		\ignore{
		A valid $O(\alpha\kappa\ws)$-regular directed flow $f$
        such that $\sum_{(i, j) \in V \times V} f_{ij}d(i, j) \ge \alpha$.
    }
\end{enumerate}
The running time is
$\widetilde{O}((rm)^{1.5} + (\kappa n)^{2})$,
where $\kappa$ is the skewness of vertex weights, $m = |E|$ and $r = \max_{e \in E} (|\Te|+|\He|)$.
\end{theorem}

Following Lemma 8 and Lemma 6 in~\cite{kale2007efficient},
two cases are analyzed, based on whether the vectors are concentrated around some vector.

\subsection{Case 1: Vectors Concentrated Case}

This is similar to~\cite[Lemma 6]{kale2007efficient}.

For vertex $i$ and radius $r$,
define $B(i, r) := \{j \in V: \norm{v_i - v_j}^2 \le r^2\}$.

We consider the case that
there exists some vertex $i_0 \in V$
such that $\omega(B(i_0, \frac{1}{\sqrt{8}\ws})) \ge \frac{\ws}{4}$.
This can be verified in time $O(n^2)$.

\begin{lemma}[Case 1 of \oracle]
\label{lemma:case1}
Suppose
there exists some vertex $i_0 \in V$
such that $\omega(B(i_0, \frac{1}{\sqrt{8}\ws})) \ge \frac{\ws}{4}$.

Then, there is an algorithm with running time
$O((rm)^{1.5})$, where $m = |E|$ and $r = \max_{e \in E} (|\Te|+|\He|)$
that outputs one of the following:
\begin{enumerate}

		\item A subset $S$ with directed sparsity
        $\vartheta(S) = \frac{w(\partial^+(S))}{\omega(S)\omega(\overline{S})}
        = O(\alpha)$.

		\item Dual variables $(z, (f_p: p \in \mathcal{T}))$ and flow matrix $\bm{F}$
		satisfying Definition~\ref{defn:oracle}.
		
		Moreover, the spectral norm satisfies $\norm{\sum_p f_p \, \bm{T}_p + z \, \bm{K} - \bm{F}} \leq O(\alpha \ws^2)$.
		
		\ignore{
		A valid $O(\alpha\kappa\ws)$-regular directed flow $f$
        such that $\sum_{(i, j) \in V \times V} f_{ij}d(i, j) \ge \alpha$.
    }
\end{enumerate}
\end{lemma}

\begin{proof}

Let $L := B(i_0, \frac{1}{\sqrt{8}\ws})$ and $R = V \backslash L$.
Following step 1 of the proof of~\cite[Lemma 5]{kale2007efficient},
we denote $\Delta(j, L) := \min_{i \in L}\norm{v_i - v_j}^2$ and obtain
\begin{align*}
1 & = \sum_{\{i, j\} \in \binom{V}{2}}\omega_i\omega_j\norm{v_i - v_j}^2 \\
& \le \sum_{\{i, j\} \in \binom{V}{2}}
\omega_i\omega_j(2\norm{v_i - v_{i_0}}^2 + 2\norm{v_{i_0} - v_j}^2)
= \sum_{i \in V} 2\omega_i(\ws - \omega_i)\norm{v_i - v_{i_0}}^2 \\
& \le \sum_{i \in V} 2\omega_i\ws(2\Delta(i, L) + 2 \cdot \frac{1}{8\ws^2})
= 4\ws\sum_{i \in V} \omega_i(\Delta(i, L) + \frac{1}{8\ws^2}).
\end{align*}

Thus, $\sum_{j \in R} \omega_j \Delta(j, L)
=
\sum_{i \in V} \omega_i\Delta(i, L) \ge \frac{1}{4\ws} - \sum_{i \in V}\frac{\omega_i}{8\ws^2}
= \frac{1}{8\ws}$.
Let $\gamma = \frac{\omega(R)}{\omega(L)}$;
note that $\gamma \le \frac{3\ws / 4}{\ws / 4} = 3$.

Now consider the two quantities
$Q_L := \sum_{i \in L} \gamma \omega_i\norm{v_0 - v_i}^2$ and
$Q_R := \sum_{j \in R} \omega_j\norm{v_0 - v_j}^2$.
We first consider the case that $Q_L \leq Q_R$.


\ignore{
We transform $H$ to a normal graph and introduce a source node $s$ and a sink $t$.
Namely, we construct a graph $G$ as follows.
\begin{itemize}
    \item Include vertex set $V$ in $G$.
    \item For each $e \in E$, create two vertices $v_e^\tail$ and $v_e^\head$.
    Add an edge from each $i \in \Te$ to $v_e^\tail$
    and from $v_e^\head$ to each $j \in v_e^\head$,
    with capacity $+\infty$.
    Add an edge from $v_e^\tail$ to $v_e^\head$ with capacity $w_e$.
    \item Create a source node $s$ and a sink node $t$.
    Add an edge from $s$ to each $i \in L$ with capacity $8\gamma\ws\omega_i\alpha$.
    Add an edge from each $j \in R$ to $t$ with capacity $8\ws\omega_j\alpha$.
\end{itemize}
}

\noindent \textbf{Max-Flow Instance in Directed Hypergraph.}
We consider the following max-flow instance $G$.  Each directed edge~$e$
in the original hypergraph $H = (V,E)$ has capacity $c_e = \frac{w_e}{2}$.

\noindent \emph{Source.} We add an extra source vertex $s$, and edges $\{(s,i): i \in L\}$,
each of which has capacity $8\gamma\ws\omega_i\alpha$.

\noindent \emph{Sink.} We add a sink vertex $t$, and edges
$\{(j,t): j \in R\}$, each of which has capacity $8\ws\omega_j\alpha$.

A max-flow can be computed in $G$, for instance, by
using the reduction to directed normal graph in Fact~\ref{fact:reduction}.


\noindent \textbf{Case A.} Suppose the flow does not saturate all source (and sink) edges,
i.e., the flow is less than $8\ws \cdot \omega(R) \cdot \alpha$.

Let $S$ be the set of vertices in $V$
that are reachable from $s$ in the residual graph.
Denote $V_s := L \backslash S$ and $V_t := R \cap S$.
Observe that $w(\partial^+(S))
\le 8\ws\alpha(\omega(R) - \gamma\omega(V_s) - \omega(V_t))$.
As
$$\max\{\omega(S), \omega(\overline{S})\} \ge \ws / 2$$
and
$$\min\{\omega(S), \omega(\overline{S})\} \cdot \gamma
\ge \min\{\omega(L \backslash V_s), \omega(R \backslash V_t)\} \cdot \gamma
\ge \omega(R) - \gamma\omega(V_s) - \omega(V_t),$$
the algorithm returns $S$ with
$\vartheta(S) = \frac{w(\partial^+(S))}{\omega(S)\omega(\overline{S})}
\le O(\alpha)$, as required.

\medskip

\noindent \textbf{Case B.} Suppose the max flow saturates all edges from $s$
(and thus also saturates all edges going into $t$).
The max flow induces a flow $\mf{f}$ in the original graph,
by ignoring the newly added edges.
Let $\bm{F}$ be the resulting flow matrix as in Definition~\ref{defn:flow},
and consider
the corresponding flow decomposition $\bm{F} := \sum_{p \in \mathcal{T}} f_p \, \bm{T}_p + \bm{D}$
as in Definition~\ref{defn:flow_decomp},
where each non-zero demand $\mf{d}_{ij}$ in $\bm{D}$ must be from $(i,j) \in L \times R$.

The algorithm returns dual variable $(z = \alpha, (f_p: p \in \mathcal{T}))$.
It suffices to check that the conditions in Definition~\ref{defn:oracle}
are satisfied.

Obviously, $z \geq \alpha$.  Moreover, since $\bm{F}$ respects
edge capacities, by Fact~\ref{fact:cflow},
$\bm{F} \bullet \bm{X}^* \leq \frac{1}{2}\sum_{e \in E} w_e d^*_e$,
for any $\bm{X}^* \succeq 0$ with associated directed distance $d^*$.
From Definition~\ref{defn:flow},
$\bm{F}$ is a linear combination of $\bm{A}_{ij}$,
each of which has $\one$ as a $0$-eigenvector.

We next verify the condition
$(\sum_{p \in \mathcal{T}} f_p \bm{T}_p + z \bm{K} ) \bullet \bm{X} \le \bm{F} \bullet \bm{X}$.
Since the candidate matrix $\bm{X}$ satisfies
$\bm{K} \bullet \bm{X} = 1$,
this reduces to
$\bm{D} \bullet \bm{X} = \sum_{i \in L, j \in R} \mf{d}_{ij} d(i,j) \geq \alpha$.

The saturation condition implies that
for each $i \in L$, $\sum_{j \in R} \mf{d}_{ij} = 8\gamma\ws\omega_i\alpha$;
similarly, for each $j \in R$,
$\sum_{i \in L} \mf{d}_{ij} = 8\ws\omega_j\alpha$.

Therefore, we have

$$\sum_{(i, j) \in L \times R} \mf{d}_{ij}\norm{v_i - v_j}^2
\ge \sum_{j \in R}8\ws\omega_j\alpha \cdot \Delta(j, L)
\ge \alpha,$$
which implies $\sum_{(i, j) \in L \times R} \mf{d}_{ij} (d(i, j) + d(j, i)) \ge 2\alpha$,
since $d(i, j) + d(j, i) = 2\norm{v_i - v_j}^2$, $\forall i, j \in V$.

On the other hand,
\begin{align*}
\sum_{(i, j) \in L \times R}\mf{d}_{ij}(d(i, j) - d(j, i))
& = \sum_{(i, j) \in L \times R}2 \mf{d}_{ij}\norm{v_0 - v_j}^2
- \sum_{(i, j) \in L \times R}2 \mf{d}_{ij}\norm{v_0 - v_i}^2 \\
& = \sum_{j \in R} 2 \cdot 8\ws\omega_j\alpha\norm{v_0 - v_j}^2
- \sum_{i \in L} 2 \cdot 8\gamma\ws\omega_i\alpha\norm{v_0 - v_i}^2 \\
& = 16\alpha\ws(Q_R - Q_L) \ge 0,
\end{align*}
by the assumption $Q_L \leq Q_R$.
Thus, we conclude that
$\sum_{(i, j) \in L \times R}\mf{d}_{ij}d(i, j) \ge \alpha$, as required.

\noindent \emph{Bound on Spectral Norm.}
Observe that
$\norm{\sum_p f_p \, \bm{T}_p + z \, \bm{K} - \bm{F}} = \norm{\alpha \, \bm{K} - \bm{D}}
\leq \alpha \, \norm{\bm{K}} + \norm{\bm{D}} \leq
\frac{\alpha \kappa \ws^2}{n} + O(\alpha \ws^2)$,

\noindent where the bound for $\norm{\bm{K}}$ comes from Corollary~\ref{cor:K},
and $\norm{\bm{D}} \leq O(\sum_{ij} \mf{d}_{ij})$ comes from
Fact~\ref{fact:dem_norm}.
Assuming $\kappa \leq n$, the spectral norm is at most $O(\alpha \ws^2)$, as required.

If $Q_L =\sum_{i \in L} \gamma\omega_i\norm{v_0 - v_i}^2 > Q_R =
\sum_{j \in R} \omega_j\norm{v_0 - v_j}^2$, then
we just reverse the directions of all edges touching $s$ or $t$ in $G$
and compute the max-flow from $t$ to $s$.
The argument is analogous.

\noindent \textbf{Running time.} The most expensive step,
a max-flow computation in a directed (normal) graph with $O(rm)$ edges,
which can be done in $\widetilde{O}((rm)^{1.5})$ time
using the algorithm of~\cite{GoldbergR98}.
Hence,  the running time is $\widetilde{O}((rm)^{1.5})$.
\qed \end{proof}

\subsection{Case 2: Vectors Well-Spread Case}

This case is similar to~\cite[Lemma 6]{kale2007efficient}.
We will use~\cite[Lemma 14]{kale2007efficient}
and~\cite[Lemma 7]{kale2007efficient},
which we state below without proof.

\begin{lemma}[Lemma 14 of \cite{kale2007efficient}]\label{lem:14}
Suppose $|V| = n$ and
each vertex $i \in V$ is associated with a vector $v_i \in \R^n$ such that
$\norm{v_i}^2 \le 1$, $\forall i \in V$ and
$\sum_{\{i, j\} \in \binom{V}{2}}\norm{v_i - v_j}^2 \ge an^2$
for some $a > 0$.
Then for at least $\frac{a}{32}$ fraction of directions $u$,
there exist $S, T \subseteq V$, each of size at least $\frac{a}{128}$,
such that $\langle v_j - v_i, u\rangle \ge \frac{a}{48\sqrt{n}}$, $\forall (i, j) \in S \times T$.
\end{lemma}

Vertex pair $(i, j)$ is said to be a \emph{$(\eta, \sigma)$-stretched} pair
along a unit vector (also called a \emph{direction}) $u$ if
$\norm{v_i - v_j}^2 \le \frac{\eta}{\sqrt{\log n}}$ and
$\langle v_j - v_i, u\rangle \ge \frac{\sigma}{\sqrt{n}}$.

\begin{lemma}[Lemma 7 of \cite{kale2007efficient}]\label{lem:7}
Let $v_1, v_2, \dots, v_n$ be vectors of length at most $1$ such that
for a $\gamma$ fraction of directions $u$,
there is a matching of $(\eta, \sigma)$-stretched pairs along $u$
of size at least $\epsilon n$.
Let $\mu > 0$ be a given constant.
Then there is a randomized algorithm which,
in time $\widetilde{O}(n^2 + \frac{1}{\mu}kn^{1 + \mu})$,
finds $k$ vertex-disjoint paths of length at most
$\frac{2C}{\mu}\sqrt{\log n}$ such that
the $\ell^2_2$-path inequality along each path
$q = (i_0, i_1, \dots, i_l)$ is violated by at least $s$,
i.e., $\sum_{j = 1}^l\norm{v_{i_{j - 1}} - v_{i_j}}^2 - \norm{v_{i_0} - v_{i_l}}^2 \le -s$,
provided that $k \le \frac{\mu\epsilon}{4C} \cdot \frac{n}{\sqrt{\log n}}$
and $\eta \le \frac{\mu s}{4C}$.
Here $s$ and $C$ are constants that depend on
$\gamma$, $\epsilon$ and $\sigma$ only.
\end{lemma}

\noindent \textbf{Pre-processing.}
In this case, for all $i \in V$,
$\omega(B(i, \frac{1}{\sqrt{8}\ws})) < \frac{\ws}{4}$.

First, we claim that there is a vertex $i_0$
such that $\omega(B(i_0, 3 / \ws)) \ge \ws / 2$,
as otherwise for all vertices $i \in V$,
there are vertices $j \in V$ with total weight greater than $\ws / 2$
such that $\norm{v_i - v_j}^2 > 9 / \ws^2$, which implies that
$$\sum_{\{i, j\} \in \binom{V}{2}}\omega_i\omega_j\norm{v_i - v_j}^2
= \frac{1}{2}\sum_{i \in V}\sum_{j \ne i}\omega_i\omega_j\norm{v_i - v_j}^2
> \frac{1}{2} \cdot \ws \cdot \frac{\ws}{2} \cdot \frac{9}{\ws^2} > 1,$$
a contradiction. Let $S := B(i_0, 3 / \ws)$.
For every $i \in S$, since
$\omega(B(i, \frac{1}{\sqrt{8}\ws})) < \ws / 4$,
we conclude that
$\omega\{S \backslash B(i, \frac{1}{\sqrt{8}\ws})\}
> \ws / 2 - \ws / 4 = \ws / 4$ and thus
$$\sum_{\{i, j\} \in \binom{S}{2}}\omega_i\omega_j\norm{v_i - v_j}^2
> \frac{1}{2} \cdot \frac{\ws}{2} \cdot \frac{\ws}{4} \cdot \frac{1}{8\ws^2}
= \Omega(1).$$
Therefore, we obtain a subset $S$ and a vertex $i_0 \in S$
such that
\begin{itemize}
    \item $\omega(S) \ge \Omega(\ws)$;
    \item $\norm{v_{i_0} - v_i}^2 \le \frac{9}{\ws^2}$, $\forall i \in S$;
    \item $\sum_{\{i, j\} \in \binom{S}{2}}\omega_i\omega_j\norm{v_i - v_j}^2 \ge \Omega(1)$.
\end{itemize}

\begin{lemma}[Case 2 of \oracle]
Suppose after the pre-processing step,
we have obtained subset $S \subseteq V$
and $i_0 \in S$ as described above.

Then, there is an $\widetilde{O}((rm)^{1.5}+ (\kappa n)^{2})$-time algorithm
that outputs one of the following:
\begin{enumerate}
    \item A subset $S'$ with directed sparsity $\vartheta(S') = O(\sqrt{\log \kappa n}) \cdot \alpha$.
		
		\item Dual variables $(z, (f_p: p \in \mathcal{T}))$ and flow matrix $\bm{F}$
		satisfying Definition~\ref{defn:oracle}.
		
		Moreover, the spectral norm satisfies $\norm{\sum_p f_p \, \bm{T}_p + z \, \bm{K} - \bm{F}} \leq O(\alpha \ws^2 \sqrt{\log \kappa n})$.

    %
				%
				
\end{enumerate}
\end{lemma}

\begin{proof}
Define $\widehat{v}_i := \frac{\ws}{3}(v_i - v_{i_0})$ for each $i \in V$.
Thus we have:

%

\begin{itemize}
    \item $\omega(S) \ge \Omega(\ws)$,
		\item  $\widehat{v}_{i_0} = 0$,
    \item $\norm{\widehat{v}_i}^2 \le 1$, $\forall i \in S$, and
    \item $\sum_{\{i, j\} \in \binom{S}{2}}\omega_i\omega_j\norm{\widehat{v}_i - \widehat{v}_j}^2 \ge \Omega(\ws^2)$.
\end{itemize}

We treat vertex $i$ with weight $\omega_i \in \{1, 2, \ldots, \kappa\}$
as $\omega_i$ identical copies, each of which has unit weight.
They form a multiset $\widehat{V}$ with $|\widehat{V}| = \ws$.
Applying Lemma~\ref{lem:14} on $\widehat{V}$,
we know that with constant possibility,
for some appropriate constants $c, \sigma > 0$,
we can obtain a unit vector $u$
and subsets $L_0, R_0 \subseteq S$ each of weight at least $c\ws$
such that $\langle \widehat{v}_j - \widehat{v}_i, u\rangle \ge \frac{\sigma}{\sqrt{\ws}}$,
$\forall i \in L_0, j \in R_0$.
%
%

Let $r$ be the median distance from $\widehat{v}_0$ to the vectors
$\{\widehat{v}_i: i \in L_0\}$,
concerning the weight distribution $\{\omega_i\}$.

Define $L_0^+ := \{i \in L_0: \norm{\widehat{v}_i - \widehat{v}_0} \ge r\}$,
$L_0^- := \{i \in L_0: \norm{\widehat{v}_i - \widehat{v}_0} \le r\}$,
$R_0^+ := \{i \in R_0: \norm{\widehat{v}_i - \widehat{v}_0} \ge r\}$ and
$R_0^- := \{i \in R_0: \norm{\widehat{v}_i - \widehat{v}_0} \le r\}$.

If $\omega(R_0^+) \ge \omega(R_0^-)$, set $L = L_0^-$ and $R = R_0^+$;
else set $L = R_0^-$ and $R = L_0^+$.
This guarantees that both $L$ and $R$ have weight
at least $\frac{c\ws}{2}$ and $\forall (i, j) \in L \times R$,
$$\widehat{d}(i, j) := \norm{\widehat{v}_i - \widehat{v}_j}^2
- \norm{\widehat{v}_i - \widehat{v}_0}^2
+ \norm{\widehat{v}_j - \widehat{v}_0}^2
\ge \norm{\widehat{v}_i - \widehat{v}_j}^2.$$

\noindent \emph{Max-Flow in Directed Hypergraph.}
Define $\beta := \frac{32C}{9\mu sc}$,
where $C = C(\gamma, \frac{c}{64}, \sigma)$
and $s = s(\gamma, \frac{c}{64}, \sigma)$
are constants determined from Lemma~\ref{lem:7}.
We add a source $s$ and connect an edge
from $s$ to each $i \in L$ with capacity
$\beta\ws\sqrt{\log\ws}\omega_i\alpha$
and similarly, add a sink $t$ and connect
an edge from each $j \in R$ to $t$ with
capacity $\beta\ws\sqrt{\log\ws}\omega_j\alpha$.
Each edge~$e$ in the original hypergraph
has capacity $c_e = \frac{w_e}{2}$.
Again, a max-flow
can be computed in $\tilde{O}((rm)^{1.5})$ time.


We consider cases based on the total value flow out of the source.
By removing the extra edges,
we can form a flow matrix $\bm{F}$ in the original graph,
and consider
the flow decomposition $\bm{F} = \sum_{p \in \mathcal{T}} f_p \, \bm{T}_p
+ \bm{D}$ as in Definition~\ref{defn:flow_decomp}.
Observe that the value of the flow
is $\sum_{i \in L, j \in R} \mf{d}_{ij}$,
where the demands $\mf{d}_{ij}$ are determined by
the demand matrix $\bm{D}$.

\medskip

\noindent \textbf{Case A.} Suppose the total flow is less than
$\frac{c\beta\ws^2}{4}\sqrt{\log\ws} \cdot \alpha$,
and let $S' \subseteq V$ be the corresponding
induced min-cut connecting to the source.

The total weight of the saturated edges touching the source (or sink) is at most
$\frac{c\ws}{4}$
so there are at least $\Omega(\ws)$ weight of vertices
on both sides of the cut.

Thus, the algorithm returns $S'$, whose directed sparsity
$\vartheta(S') \leq O(\sqrt{\log\ws}) \cdot \alpha
\leq O(\sqrt{\log(\kappa n)}) \cdot \alpha$, as required.

\medskip

\noindent \textbf{Case B.}
Suppose $\bm{D} \bullet \bm{X} \geq \alpha$.
In this case, the algorithm returns
the dual solution $(z = \alpha, (f_p: p \in \mathcal{T}))$
and the flow matrix~$\bm{F}$, which
satisfies Definition~\ref{defn:oracle},
as in the proof of Lemma~\ref{lemma:case1}.

Again, the spectral norm
is bounded by $\norm{z \, \bm{K} - \bm{D}}
\leq O(\alpha \ws^2 \sqrt{\log \ws})$,
because it is dominated by $\norm{\bm{D}} \leq O(\sum_{i \in L, j \in R} \mf{d}_{ij})
\leq O(\alpha \ws^2 \sqrt{\log \ws})$.

\ignore{

We check if
$\sum_{(i, j) \in L \times R}f_{ij}d(i, j) \ge \alpha$.
If it is, then we are done. The flow is
$O(\kappa\ws\sqrt{\log\ws}\alpha)$-regular because for each $i$,
$\ws\sqrt{\log\ws}\omega_i\alpha = O(\kappa\ws\sqrt{\log\ws}\alpha)$.
}

\medskip

\noindent \textbf{Case C.}
In the remaining case,
the total flow is at least
$\frac{c\beta\ws^2}{4}\sqrt{\log\ws} \cdot \alpha$
but $\bm{D} \bullet \bm{X} < \alpha$.
We try to find a path along which the path inequality is drastically violated.

Observe that
$\sum_{i \in L, j \in R} \mf{d}_{ij} d(i,j) = \bm{D} \bullet \bm{X} < \alpha$.
On the other hand, the value of the flow is
$\sum_{i \in L, j \in R} \mf{d}_{ij} \ge \frac{c\beta\ws^2}{4}\sqrt{\log\ws} \cdot \alpha$.
Hence, by Markov's Inequality,
among the total flow
$\sum_{i \in L, j \in R} \mf{d}_{ij}$,
at least half of it is routed between pairs $(i,j)$
such that
$\frac{8}{c\beta\ws^2\sqrt{\log\ws}} \geq d(i,j) \geq \norm{v_i - v_j}^2$,
where the last inequality follows
from the choice of $L$ and $R$.
Equivalently, at least half of the total flow is routed between pairs
$(i, j)$ such that
$\norm{\widehat{v}_i - \widehat{v}_j}^2 \le \frac{8}{9c\beta\sqrt{\log\ws}}$.

Following the arguments in~\cite[Lemma 15]{kale2007efficient},
we can show that there is a matching in $\widetilde{V}$
of size at least $\frac{c\ws}{64}$, such that each matched pair
is copied from some $(i, j) \in L \times R$ with
$\norm{\widehat{v}_i - \widehat{v}_j}^2 \le \frac{8}{9c\beta\sqrt{\log\ws}}$.

Let $\gamma'$ be the fraction of directions $u$ such that
there is such a matching.
Note that $\gamma'$ is determined once all the vectors are given,
while we do not need to compute its exact value.
If $\gamma' \ge \gamma / 2$, then
after trying $O(\log n)$ random directions we will,
with high probability,
end up in Case A or B, or successfully find a matching.
If $\gamma' < \gamma / 2$, then each trial will result in Case A or B
with probability at least $\gamma - \gamma' \ge \gamma / 2$,
so trying $O(\log n)$ random directions will make \oracle end up in
Case A or B with high probability.

Now assume that the matching mentioned above exists.
The choice of $\beta$ guarantees that
$\eta = \frac{8}{9c\beta} = \frac{\mu s}{4C}.$
We apply Lemma~\ref{lem:7} on
$\omega_i$ copies of $\widehat{v}_i$, $\forall i \in V$,
with parameter $\mu = 1$ and $k = 1$.
So in $\widetilde{O}((\kappa n)^2)$ time, we can find a path
$q = (i_0, i_1, \dots, i_k)$
whose $\ell^2_2$-path inequality is violated
by at least $s$.
Specifically, define $\mathcal{T}_q := \{\tri{i_{0}}{i_j}{i_{j+1}} : 1 \leq j \leq k-1\}$.
Then, the violation condition is
$\sum_{p \in \mathcal{T}_q} \bm{T}_p \bullet \bm{X} \leq -\frac{9s}{\ws^2}$.

Next, we define the dual solution returned.
We set $z = \alpha$ and $\bm{F} = 0$.

For each $p \in \mathcal{T}_q$,
we set $f_p := \frac{\ws^2\alpha}{9s}$;
set other $f_p$'s to 0.
%

We next check that the conditions in Definition~\ref{defn:oracle}
are satisfied.  Most of them are straightforward.
In particular,

$$\left(\sum_{p \in \mathcal{T}} f_p \bm{T}_p + z \bm{K}\right) \bullet \bm{X}
=\frac{\ws^2\alpha}{9s}  \sum_{p \in \mathcal{T}_q}
\bm{T}_p \bullet \bm{X} + \alpha
\leq \frac{\ws^2\alpha}{9s} \cdot (-\frac{9s}{\ws^2}) + \alpha = 0.$$

\noindent \emph{Bound on Spectral Norm.}
Observe that for the path~$q$, the degree of each vertex is at most 2 in the path.
Hence, $\norm{\sum_{p \in \mathcal{T}_q} \bm{T}_p} \leq O(1)$.
Therefore, we have
$$\norm{\sum_{p \in \mathcal{P}} f_p \bm{T}_p + z \bm{K} - \bm{F}}
\leq \frac{\ws^2\alpha}{9s} \cdot O(1) + \alpha \cdot \frac{\kappa \ws^2}{n}
\leq O(\alpha\ws^2),$$
assuming $\kappa \le n$.

\noindent \textbf{Analysis of Running Time.}
The max-flow computation
takes $\widetilde{O}((rm)^{1.5})$ time.
Running the algorithm in Lemma~\ref{lem:7}
takes $\widetilde{O}((\kappa n)^2)$ time,
since we only need to find $1$ violating path, instead of
$\Theta(\frac{n}{\sqrt{\log n}})$ as indicated in~\cite{kale2007efficient}.
Hence, the running time is $\widetilde{O}((rm)^{1.5} + (\kappa n)^2)$.
\qed \end{proof}

{
\bibliography{dihyper}

\begin{thebibliography}{ACMM05}

\bibitem[AAC07]{stoc/AgarwalAC07}
Amit Agarwal, Noga Alon, and Moses Charikar.
\newblock Improved approximation for directed cut problems.
\newblock In {\em {STOC}}, pages 671--680. {ACM}, 2007.

\bibitem[ACMM05]{stoc/AgarwalCMM05}
Amit Agarwal, Moses Charikar, Konstantin Makarychev, and Yury Makarychev.
\newblock O(sqrt(log n)) approximation algorithms for min uncut, min 2cnf
  deletion, and directed cut problems.
\newblock In {\em {STOC}}, pages 573--581. {ACM}, 2005.

\bibitem[AK16]{AroraK16}
Sanjeev Arora and Satyen Kale.
\newblock A combinatorial, primal-dual approach to semidefinite programs.
\newblock {\em J. {ACM}}, 63(2):12:1--12:35, 2016.

\bibitem[ALN08]{arora2008euclidean}
Sanjeev Arora, James Lee, and Assaf Naor.
\newblock Euclidean distortion and the sparsest cut.
\newblock {\em Journal of the American Mathematical Society}, 21(1):1--21,
  2008.

\bibitem[ARV09]{jacm/AroraRV09}
Sanjeev Arora, Satish Rao, and Umesh~V. Vazirani.
\newblock Expander flows, geometric embeddings and graph partitioning.
\newblock {\em J. {ACM}}, 56(2), 2009.

\bibitem[CLTZ18]{chan2018jacm}
T.{-}H.~Hubert Chan, Anand Louis, Zhihao~Gavin Tang, and Chenzi Zhang.
\newblock Spectral properties of hypergraph laplacian and approximation
  algorithms.
\newblock {\em J. {ACM}}, 65(3):15:1--15:48, 2018.

\bibitem[GLPN93]{gallo1993directed}
Giorgio Gallo, Giustino Longo, Stefano Pallottino, and Sang Nguyen.
\newblock Directed hypergraphs and applications.
\newblock {\em Discrete applied mathematics}, 42(2):177--201, 1993.

\bibitem[GR98]{GoldbergR98}
Andrew~V. Goldberg and Satish Rao.
\newblock Beyond the flow decomposition barrier.
\newblock {\em J. {ACM}}, 45(5):783--797, 1998.

\bibitem[Kal07]{kale2007efficient}
Satyen Kale.
\newblock {\em Efficient algorithms using the multiplicative weights update
  method}.
\newblock Princeton University, 2007.

\bibitem[KVV04]{jacm/KannanVV04}
Ravi Kannan, Santosh Vempala, and Adrian Vetta.
\newblock On clusterings: Good, bad and spectral.
\newblock {\em J. {ACM}}, 51(3):497--515, 2004.

\bibitem[LM14]{louis2014approximation}
Anand Louis and Konstantin Makarychev.
\newblock Approximation algorithm for sparsest \emph{k}-partitioning.
\newblock In {\em {SODA}}, pages 1244--1255. {SIAM}, 2014.

\bibitem[LR99]{LeightonR99}
Frank~Thomson Leighton and Satish Rao.
\newblock Multicommodity max-flow min-cut theorems and their use in designing
  approximation algorithms.
\newblock {\em J. {ACM}}, 46(6):787--832, 1999.

\bibitem[MMV15]{colt/MakarychevMV15}
Konstantin Makarychev, Yury Makarychev, and Aravindan Vijayaraghavan.
\newblock Correlation clustering with noisy partial information.
\newblock In {\em {COLT}}, volume~40 of {\em {JMLR} Workshop and Conference
  Proceedings}, pages 1321--1342. JMLR.org, 2015.

\bibitem[PSZ15]{PengSZ15}
Richard Peng, He~Sun, and Luca Zanetti.
\newblock Partitioning well-clustered graphs: Spectral clustering works!
\newblock In {\em {COLT}}, volume~40 of {\em {JMLR} Workshop and Conference
  Proceedings}, pages 1423--1455. JMLR.org, 2015.

\end{thebibliography}
\bibliographystyle{alpha}
}


\end{document}